\documentclass[a4paper,11pt,twoside]{preprint}
\usepackage{sub_JP}
\usepackage{times}
\usepackage{amsmath}
\usepackage{amssymb}
\usepackage{enumerate}
\usepackage{fancyhdr}
\renewcommand\refname{References}
\usepackage{amsthm}
\usepackage{mathrsfs}
\let\rho=\varrho
\let\phi=\varphi
\def\real{\mathbb R}
\newtheorem{theorem}{Theorem}[section]
\newtheorem{lemma}[theorem]{Lemma}

\def\ie{{\it i.e.}}
\def\eg{{\it e.g.}}
\def\HALF{{\textstyle\frac{1}{2}}}

\newenvironment{remark}{\noindent\textbf{Remark.}}{}

\def\HALF{{\textstyle\frac{1}{2}}}

\usepackage[scanall]{psfrag}
\usepackage{epsfig}
\def\C#1{\begin{center}{#1}\end{center}}
\def\fref#1{Fig.~\ref{#1}}
\def\tref#1{Table~\ref{#1}}
\def\eref#1{(\ref{#1})}
\def\sqq{\sqrt{1-\epsilon ^2}}

\def\sq{\sqrt{\cdot}}

\def\pmin{p_{\rm min}}
\def\pmax{p_{\rm max}}

\def\OO{{\cal O}}

\def\stwo{{\sqrt{2}}}
\def\emph#1{{\it #1}}

\def\d{{\rm d}}

\let\epsilon=\varepsilon
\def\R{{\rm R}}

\def\L{{\rm L}}

\def\d{{\rm d}}

\def\eref#1{(\ref{#1})}
\def\sref#1{Sect.~\ref{#1}}

\usepackage{MHequ}

\def\pmax{p_{\rm max}}
\def\IL{I_\L}
\def\IR{I_\R}
\def\ILeff{I_{\rm L,eff}}
\def\IReff{I_{\rm R,eff}}

\begin{document}
\title{Rattling and freezing in a 1-D transport model}
\bigskip
\author{Jean-Pierre Eckmann${}^1$ and Lai-Sang Young${}^2$}
\institute{
${}^1$D\'epartement de Physique Th\'eorique et Section de Math\'ematiques,
Universit\'e de Gen\`eve, CH-1211 Gen\`eve 4 (Switzerland)\\
${}^2$Courant Institute of Mathematical Sciences\\
New York University, NY 10012 (USA)
}
\maketitle
\thispagestyle{empty}

\tableofcontents
\bigskip
\begin{abstract}  We consider a heat conduction model introduced in \cite{CE2009}.
This is an open system in which particles exchange momentum with a row 
of (fixed) scatterers. We assume simplified bath conditions throughout, and give
a qualitative description of the dynamics extrapolating from the case of a single
particle for which we have a fairly clear understanding. The main phenomenon
discussed is {\it freezing}, or the slowing down of particles with time. As particle
number is conserved, this means fewer collisions per unit time, and less contact
with the baths; in other words, the conductor becomes less effective. Careful 
numerical documentation of freezing is provided, 
and a theoretical explanation is proposed. Freezing being an extremely slow 
process, however, the system behaves as though it is in a steady state for
long durations. 
Quantities such as energy and fluxes are studied, and are found to have 
curious relationships with particle density.
\end{abstract}

\bigskip
\section{Introduction}
In this paper, we report on a study of a model introduced in 
\cite{CE2009} and further studied in \cite{CEM2009}.
In this model of 1-dimensional scattering, there is a row of scatterers, 
which are equally spaced and tied down, the interval between adjacent
scatterers being called a ``gap".  Each scatterer carries 
a single variable, namely its ``momentum", called $q_i$. Moving
about in the system are particles whose momenta and positions
are denoted by $p_j$ and $x_j$ respectively. The particles do not
see each other directly, and interactions between
particles and scatterers follow the rules of elastic scattering --
except that the scatterers do not move. 
The two ends of the chain are connected to heat baths. Unlike \cite{CE2009,CEM2009}, which considered the grand canonical case,
the \emph{canonical} 
case is studied in this paper, \ie, the number of particles is fixed,
and when a
particle reaches a bath, a new one is injected immediately with a
velocity chosen from some distribution.  Details of this model are
given in \sref{s:model}.

The fact that the scatterers have {\it no recoil}, \ie, their 
positions are fixed, simplifies 
the analysis, making the local dynamical picture more tractable. As we will
show, locally the dynamics are characterized by {\it rattling}: 
a particle will rattle back 
and forth between two scatterers, 
carrying momentum from one to the other until a certain state is
reached, at which time it exits this interval (passing through one of
the scatterers), and begins to rattle between the two scatterers in 
its new interval. 

Global problems such as transport and asymptotic states of the chain
are beyond the reach of rigorous mathematics. Using bath distributions
concentrated on two short intervals, \ie, injected particles have momenta
$\sim \IL $ for the left bath and $\sim \IR $ for the right, we have a fairly complete description of the dynamics in the case of a single particle mediating the transfer or momentum among an arbitrary number of scatterers. This description
is based on a combination of rigorous and heuristic arguments, 
and the results are confirmed in simulations.
Extrapolating from this single-particle case, we believe we also have
a reasonable understanding in the case where particle density is
low. As particle density increases, the situation becomes  
very complex. We will report on 
a number of observations and numerical findings, some of which 
we must admit we had not anticipated at the outset.

Our original intent was to study the nonequilibrium steady states of
this model (with Maxwellian bath distributions), which we had 
thought would resemble steady states in  similar models such 
as \cite{EY2006,LY2010} or \cite{BLY2010}; for a review of the 
broader subject, see \cite{Lepri2003}. 
What we found instead was that the model considered in this paper
{\it freezes}, \ie, it slows down with time: Through
its exchanges with the scatterers, each particle
acquires, from time to time, very low values of momentum. 
Once such a low momentum is acquired, the particle
spends an extremely long time traveling between
two scatterers, which are unit distance apart. During this time it has
no influence on the evolution of the system. We find that as time
goes on, more and more particles are stuck in these low energy states,
with fewer and fewer collisions occurring per unit time.\footnote{Freezing here 
refers to the slowing down of particles, not to falling of scatterer temperatures.
On the contrary, mean scatterer temperature
climbs slowly with time, as will be explained in
\sref{s:numevidence}} 
We further conjecture -- not without reason --
that ``at the end of time", the action in the entire system is carried, 
for the most part, by a single particle. That is to say, for arbitrarily
long durations, the dynamics are those in the single-particle case. 

We will provide numerical documentation of freezing;
that takes quite a bit of computing time since 
the process is gradual and very slow. We will also connect the no-recoil property of scatterers directly to freezing in one physical dimension,
thereby providing a theoretical basis for understanding this phenomenon. More precisely, we will argue that
at least in closed systems (no heat baths), {\it on balance
particle energy is dissipated
through collisions with scatterers that do not recoil}.
In modeling, it is not uncommon to accept an unphysical property
to make a model more analyzable. In this case we have found that
the no-recoil property has an unintended consequence.

While freezing is the main message of this paper, we would argue
 that quantities such as total scatterer energy and fluxes in and
out of the chain -- observed in real time -- are entirely relevant. 
This is because the freezing process is extremely slow: We have 
found that following a relatively brief initial transient,
the system will settle down to what is perhaps best described as a
{\it quasi-stationary state}, \ie, a very-slowly-varying state which for
many purposes can be treated as stationary. We have found that
during this infinitely long period of quasi-stationarity, the chain exhibits 
macroscopic behaviors that are very challenging to explain. 

\medskip
This paper is organized as follows: In Section 2, we
describe the model and explain the mechanics of rattling. 
Section 3 discusses
transport properties of the chain assuming quasi-stationarity, and 
Section 4 is devoted to the documentation and discussion of freezing.

%%%%%%%%%%%%%%%%%%%%%%%%%%%%%%%%%%
%%%%%%%%%%%%%%%%%%%%%%%%%%%%%%%%%
\section{Model and local dynamics}\label{s:coordinates}

\subsection{Model description}\label{s:model}
We consider $N$ scatterers at positions $i=1,\dots,N$ and a ``bath" each
at positions $0$ and $N+1$. The $i$th scatterer has a momentum
$q_i\in\real$. There are $n$ particles in the system; the momentum of
the $j$th particle is denoted by $p_j\in\real$. 
Each particle moves with uniform velocity until it reaches
either a scatterer or a bath.
A collision between a particle and a scatterer or bath
is called an {\it event} in this paper.

In a particle-scatterer collision, momentum
is exchanged as follows: For $p=p_j$ and $q=q_i$,
the ``scattering'' replaces these values by $p'$ and $q'$ given by
\begin{equa}
  p' &= -\sigma p + (1-\sigma) q~,\\
  q' &= (1+\sigma) p+ \sigma q~.\\
\end{equa}
In these coordinates, which were used in \cite{CE2009}, the
scattering matrix is
\begin{equ}
  \begin{pmatrix}
    -\sigma& 1-\sigma\\
1+\sigma & \sigma\\
  \end{pmatrix}\ ,
\end{equ}
and the energy  
\begin{equ}
  E=\,p^2+\frac{1-\sigma}{1+\sigma}q^2
\end{equ}
is preserved. In the present paper, it is more convenient to 
rescale $p$ and $q$ in such a way that $(p,q)$ lies on a circle: 
\begin{equ}
 \hat p= p~,\qquad \hat q =\sqrt{\frac{1-\sigma }{1+\sigma }}q~.
\end{equ}
In these new coordinates, the matrix takes the form
\begin{equ}
  \begin{pmatrix}
    -\sigma & \sqrt{1-\sigma^2}\\
\sqrt{1-\sigma^2}& \sigma\\
  \end{pmatrix}~.
  \end{equ}
We will be mostly interested in $\sigma\approx 1$ and thus introduce
$\epsilon =\sqrt{1-\sigma^2}$. 
In these coordinates (and we will omit the ``hat'' from now on), the
energy is $p^2+q^2$, and the scattering matrix is
\begin{equ}
S=  \begin{pmatrix}
    -\sqrt{1-\epsilon ^2} & \epsilon \\
\epsilon &
\sqrt{1-\epsilon ^2}\\
  \end{pmatrix}~,
\end{equ}
which geometrically can be seen as a reflection $p \mapsto -p$ 
followed by a clockwise rotation by an angle 
$\theta = -\arctan(\epsilon/\sqrt{1-\epsilon ^2})$.

Once the particle has scattered, it continues on its way, moving to
the right of the scatterer if $p>0$ and to the left if $p<0$.

When a particle
reaches a ``bath", it disappears and is instantaneously replaced by 
one with a new energy. 
We have used the word ``bath" here to denote an infinite source of energy
or momenta. No Maxwellian character is assumed, as that will be
destroyed by the no-recoil property of the scatterers in any case.
For simplicity, we fix two positive numbers, $\IL$ and $\IR$, which are 
nominal values at the bath, and introduce a small amount of
randomness by fixing a (small) parameter 
$\alpha >0$. The momenta of injected particles are 
{\it i.i.d.}~random variables with a uniform
distribution on $[\IL(1-\alpha ),\IL(1+\alpha )]$ for the left bath, respectively
$[-\IR(1+\alpha ),-\IR(1-\alpha )]$ for the right bath. In
general, unless stated otherwise, simulations have been done
with
\begin{equ}
\IL= 1~,\quad \IR= 3~,\quad\epsilon =0.05,\quad {\rm and} \quad \alpha =0.1~,
\end{equ}
and $400$ scatterers.
It will be convenient to define
\begin{equ}
  \ILeff= \IL(1-\alpha )~,\qquad  \IReff= \IR(1-\alpha )~,
\end{equ}
because, as we shall see, these are the values that occur most
prominently when $\epsilon $ is close to $0$.

%%%%%%%%%%%%%%%%%%%%%%%%%%%%
\subsection[Rattling]{Rattling}\label{s:rattling}

The mechanism by which momentum is transported between adjacent
scatterers can be
summarized as follows: 
\begin{itemize}
\item[(1)] {\it Conditions for ``crossing".} When a particle collides with
a scatterer, it can be reflected back, or it can pass through the scatterer;
the latter is the definition of ``crossing". 
If a particle with momentum $p>0$ collides with a scatterer with 
momentum $q<p \sq /\epsilon$, then by the rules
of scattering, it will be reflected. Similarly for $p<0$ and 
$q> p\sq /\epsilon$. In particular, crossing is impossible if $p\cdot q<0$.

\item[(2)] {\it Rattling.} Consider an interval bounded by two scatterers with 
momenta $q_1^0$ (on the left) and $q_2^0$ (on the right), and assume
that $q_1^0>0$ and $q_2^0<0$. Suppose a particle has just
entered this interval. The particle will  rattle back and forth 
carrying momentum from one scatterer to the other, and it will exit (on one of
the two sides) in finite time. After this crossing, the momenta of
the two scatterers are $q_1\approx q_2^0$ and $q_2\approx
q_1^0$.
\end{itemize}

\begin{figure}[t]
  \C{\psfig{file=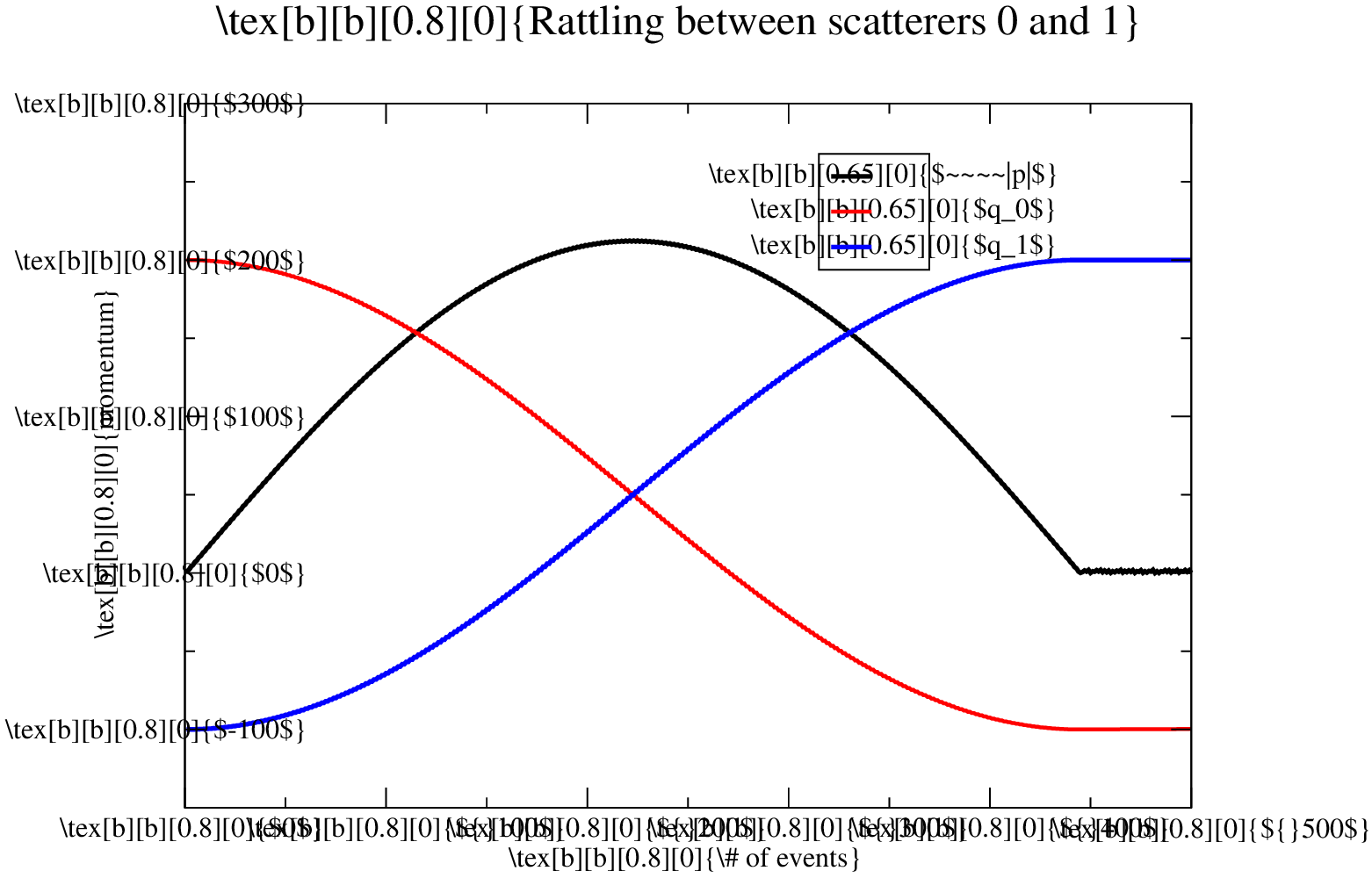,width=0.7\textwidth,angle=270}}
  \caption{Rattling between two consecutive scatterers (for $\epsilon =0.01$). The horizontal
    axis is events. The vertical axis is momentum for $q_1$ and $q_2$
    and $|p|$. (The sign of $p$ switches after each collision during rattling.)
Note the exchange of the values of $q_1$ and $q_2$.}\label{f:rattling}
 \end{figure}

%%%%%%%%%%%%%%%%%%
\subsubsection{Derivation}\label{s:derivation}

We will analyze the case when a particle starts from scatterer 1,
has initial momentum $p>0$ so that it flies to the right,
rebounds from scatterer 2, then rebounds from scatterer 1 again,
and so on.
We will follow its trajectory until it exits the interval between scatterers 
1 and 2, assuming throughout that there is no other particle in the picture.

Applying
the matrix $S$ twice (first on $p$ and $q_2$ and then on $p$ and $q_1$) we get the sequence:
\begin{equ}
  \begin{matrix}
    p & \to & -p\sq +\epsilon q_2&\to & p(\sq)^2 -q_2 \epsilon \sq
    +\epsilon q_1~,\\
    q_1 & \to & q_1&\to & q_1\sq-p\epsilon\sq 
    +\epsilon^2 q_2~,\\
    q_2 & \to &q_2\sq+\epsilon p &\to & q_2\sq +\epsilon p~.\\
  \end{matrix}
\end{equ}
To discuss this sequence, it is convenient to formulate the
problem as a differential equation, which takes the form 
%(with $n$ the discrete ``time''):
\begin{equa}\label{e:diffequ}
  \frac{\d p}{\d n}&=-\epsilon (q_2-q_1)+\OO(\epsilon ^2)~,\\
  \frac{\d q_1}{\d n}&=-\epsilon p+\OO(\epsilon ^2)~,\\
  \frac{\d q_2}{\d n}&=+\epsilon p+\OO(\epsilon ^2) ~.\\
\end{equa}
Here $n$ is  ``discrete time", the passage of each
unit of which corresponds to one lap (two reflections) for the particle.
The solution is then, up to order 2 in $\epsilon $ with initial
conditions $p$, $q_1$, $q_2$:
\begin{equa}\label{e:diff}
  p(n)&= (q_1-q_2)\sin(\epsilon n\stwo)/\stwo+p\cos(\epsilon
  n\stwo)~,\\
  q_1(n)&= \HALF(q_1+q_2)+\HALF (q_1-q_2)\cos(\epsilon
  n\stwo)-p\sin(\epsilon n\stwo)/\stwo~,\\
  q_2(n)&= \HALF(q_1+q_2)+\HALF (q_2-q_1)\cos(\epsilon
  n\stwo)-p\sin(\epsilon n\stwo)/\stwo~.\\
\end{equa}
Observe next that if $p\approx 0$ then $p(n)$ passes again through 0 for $n\approx 
\pi/(\stwo \epsilon )$. Thus this is about the time when the particle
will leave the interval between scatterer 1 and scatterer 2 (either to
the left or to the right). We find from \eref{e:diff} and
from
$\cos(\epsilon n\stwo)=-1$ that $q_1(n)=q_2$ and $q_2(n)=q_1$. Thus,
in the case $p\approx 0$ the values of $q_1$ and $q_2$ are simply
exchanged after the rattling, as we have asserted.

%%\begin{remark}
%%  Note that at each collision, the value $p'$ after scattering
%%is
%%$p'=-p\sq +\epsilon q$ and so, $\sign~ p =\sign~p'$ iff 
%%$$
%%|p|<\epsilon |q|/\sq~.
%%$$
%%This means that a particle will only escape from rattling if $|p| $
%%is small, relative to $|q|$.
%%\end{remark}
%%
The time evolution of $(p(n), q_1(n), q_2(n))$ is illustrated in \fref{f:rattling}.

The question is now on which side the particle will exit. We claim,
and have checked numerically, that the probability $P_{\rm R}$ to
leave to the right as compared to that of leaving to the left, $P_{\rm
  L}$ satisfies approximately
\begin{equ}\label{e:exit}
 P_{\rm R}/P_{\rm L} = |q_{\rm L}| / |q_{\rm R}|~,
\end{equ}
where 
$q_{\rm L}$ is the value of the left-hand $q$ at the time of exit. We
have no proof of this,
but intuitively, this ratio can be understood as follows: For the
particle to exit, its momentum must satisfy $|p|< \epsilon |q|/\sq$,
with $p$ and $q$ of the same sign. By the study above, each collision
is a rotation by an angle of order $\OO(\epsilon )$, in the $p,q$
plane, so the probabilities to first satisfy $|p|< \epsilon |q_{\rm
  L}|/\sq$ or $|p|< \epsilon |q_{\rm R}|/\sq$ are proportional to $|q_{\rm L}| / |q_{\rm R}|$.

%%%%%%%%%%%%%%%%%%%%%%%%%%%%%%%%%%

\subsection{The baths}\label{s:bath}

We consider next how particles enter the system. For definiteness,
consider the left bath. Observe that typically, many attempts
are made before a new particle successfully crosses the first scatterer.
This is because in order for a particle to leave the chain for (say) the
left bath, it must cross the first scatterer from right to left. At the time of
this crossing, $q_1$ must necessarily be $<0$,  
and in order for a new particle 
with $p \approx \IL $ to cross it again to enter the chain, 
$q_1$ must be raised to $ \gtrsim \IL /\varepsilon$; see the first
summary item at the beginning of \sref{s:rattling}.  Raising $q_1$ to
this required value is done via a rattling mechanism similar to that discussed 
in the previous subsection, except that bath injections are all
 $\approx \IL $. The following gives a bound on the number of attempts
 before a successful entry:

\begin{lemma}\label{l:bathscatterer}
Suppose (i) particles with momenta $p\in[\pmin,\pmax]$ are injected,
and (ii) the first scatterer has initial momentum $q$ and its momentum 
is not altered by any other particle during the
period in question. Then
after at most 
\begin{equ}\label{l1}
  n_{\rm max} =\max\left(0, \frac{\pmax-\epsilon q}{\epsilon (\pmin-\pmax/2)}\right)
\end{equ}
collisions with the scatterer
the particle will cross into the system, and after crossing, 
\begin{equ}\label{l2}
\frac{\pmin}{\epsilon } \le q_{\rm after}\le \frac{5\pmax}{4\epsilon }~.
\end{equ}
\end{lemma}

\begin{proof}
If, before crossing, $p>\epsilon q/\sqq$, then the particle will be
reflected and $q$ changes to $q'=q\sqq+\epsilon p$. 
This can only happen if $q< \pmax\sqq/\epsilon $. Since also $p<\pmax$ we get
\begin{equa}
  q'&\ge\epsilon \pmin +q -q(1-\sqq)\\
&\ge q+\epsilon \pmin -\pmax (1-\sqq)\sqq/\epsilon \\
& \ge q+ \epsilon (\pmin-\pmax/2)~.
\end{equa}
This proves \eref{l1}. Then \eref{l2} follows at once from $q<
\pmax\sqq/\epsilon $, the bound on $p$ and the scattering rules.
\end{proof}

\begin{remark}
  Particles that get through have $p$
 almost exclusively at the low end of theoretically admissible range.
This is entirely expected: The $q$'s change extremely slowly, by 
$\OO(\varepsilon)$ each hit, while the $p$'s given out by the baths are
random {\it i.i.d.},
so the first $p$ to get through is naturally 
at the low end of the range.
\end{remark}

%%%%%%%%%%%%%%%%%%%%%%%%%%%%%%
\section{Transport in the chain}\label{s:transport}

\subsection{Single-particle dynamics}\label{s:single}

We have a fairly complete description of the dynamics in the case 
where there is a unique particle in the system. As we will see,
this simple dynamical system is a useful reference point 
for studying chains with low particle densities (\sref{s:lowdensity}). It will also 
be a source of information for a conjectured ``final state" for 
all systems (see \sref{s:final}).

\medskip
\begin{remark}
For simplicity, we will describe all our findings in terms of the numbers ``$1$''
  and ``$3$'', which are the injection values, left and right. Our results apply, 
  {\it mutatis mutandis}, to arbitrary 
  values of $\IL$ and $\IR$, though simulations get more difficult
  when the ratio $\IL/\IR$ is very small or large. 
\end{remark}

\medskip
We find numerically that independently of the initial state of the
chain, after a long enough transient, all $q_i$ become -- with small
variations -- either a ``$1$'' or a ``$-3$''. Here is
what we mean: When a particle enters the system, more precisely
when it crosses its first scatterer, it carries a momentum close to 
$ \ILeff$ or $-\IReff$; see \sref{s:bath}. When such
a crossing occurs, the scatterer in question generally has momentum
$\sim 1/\epsilon$ times that of the particle; see \sref{s:derivation}. We assert here 
that following a long enough transient, all $q_i$ take on
essentially these values. When it is $\sim \ILeff/\epsilon$, we call it a ``$1$";
when it is $ \sim - \IReff/\epsilon$, we call it a ``$-3$".
This being the case, we start the discussion from such a configuration.

Also, in this single-particle case, the only action is where the particle
is, and it is simpler to go by {\it events time} in lieu of real physical time.  
This is what is used in the discussion to follow.

Consider the interval between $q_i$ and $q_{i+1}$. From \sref{s:rattling}, 
we see that locally, the possibilities are limited: 

\medskip
$(q_i, q_{i+1})= (1,1)$: the particle can only enter from the left and
exit on the right;  
no change in scatterer energies upon exit; similarly

 \medskip
$(q_i, q_{i+1})= (-3,-3)$: the particle can only enter from the right
 and exit on the 
left, with no change in scatterer energies upon exit;

\medskip
$(q_i, q_{i+1})= (1,-3)$: the particle can only enter from the left
and exit on the right;  
scatterer energies are transformed from $(1,-3)$ to $(-3,1)$

\medskip
$(q_i, q_{i+1})= (-3,1)$: the particle cannot enter this interval.

\medskip
\noindent We propose to view all $1$'s as indistinguishable, as are all $-3$'s,
and to consider only {\it crossings}, \ie, when the particle moves from 
one interval to another (ignoring the rattling that occurs in between). 
Seen this way, the dynamics
of a chain are completely described by (i) the walk of the particle, and 
(ii) flipping or not of scatterer energies along the way, where the
only ``flipping" possible is from $(1,-3)$ to $(-3,1)$.
An important observation 
before going further is that {\it positive momenta, \ie, $1$'s, can only move to the
right, while negative momenta, \ie, $-3$'s, can only move to the left}.

We have found numerically (and have a proof for
the single-particle case with $\IL \not \approx \IR$) that  a particle cannot
stay in an interval forever. 
We {\it assume} for purposes of the discussion below that
after the particle transforms $(1,-3)$ to $(-3,1)$ it will exit 
with probability $3/4$  to the left and with  $1/4$ to the right.
We now argue that 

\medskip

{\it in a typical steady-state configuration, there are 
$3$ times as many $1$'s 

\quad as there are $-3$'s.}

\medskip \noindent 
The argument below is part rigorous and part heuristic, 
and the phenomenon 
is fully confirmed by numerical experimentation. Notice first the following: 

\smallskip
\noindent (a) If we define a ``move'' to be a swapping of two adjacent energies, then 
each move by a 1 is accompanied by a move by a $-3$, and vice versa.

\smallskip
\noindent (b) The number of $1$'s that enter the chain is equal to the number 
of $1$'s that exit 
(it is a steady state) which in turn is equal to the number of $-3$'s that enter (because each exit of 
a 1 is followed by the entrance of a $-3$) etc.

\smallskip
In view of (a) and (b) above, we need to show that the $-3$'s 
move along the chain $3$ times as fast as 
the $1$'s. Specifically, we will argue that  in each ``sequence of flips'' 
(to be defined), a $-3$ moves on average $3$ times
in a row, while $3$ different $1$'s move one step each.

Indeed,
consider a configuration of $1$'s and $-3$'s such as
$$
\dots -3 \ 1 \ 1 \ 1 \ 1 -3 \ 1\ 1 \ 1 \ -3 \ 1 \dots~.
$$
Suppose the particle is in $(1,-3)$, the middle $-3$. 
Then it flips it to $(-3,1)$, and exits the gap.
If it goes to the left, it is again in a $(1,-3)$, and the same scenario is 
repeated. This {\it sequence of flips} ends when either (i) the particle
exits to the right after a flip, or (ii) the $-3$ comes up against the $-3$ on the left.

If (i) occurs, we might as well think of the particle as going directly to
$(1,-3)$, the $-3$ on the right (because the only exit possible in $(1,1)$ is
to the right) and nobody has moved in the meantime. The last paragraph
then repeats itself for this new $-3$. 
If (ii) occurs, then the baton is passed to the $-3$ on the left and the
same story happens.

Now given a configuration of $1$'s and $-3$'s, imagine 
doing a walk hopping from $-3$ to $-3$.
In view of our assumption of left/right exits, if the ratio
of $1$'s to $-3$'s is $>$ $3:1$, then (i) above will occur more often,
resulting in the walker visiting the right bath more often than the left. Likewise,
when there are too few $1$'s, (ii) will happen more often, leading to more
frequent visits with the left bath. As noted earlier, the two baths
are visited with equal frequency in a steady state.  Hence
in a steady state, there are roughly $3$ times as many $1$'s as there are $-3$'s.  

\medskip
Results of numerical verifications for various assertions are
shown in \tref{t:table1}. 

%% results made with file gethits.pl
\begin{table}
\begin{center}  \begin{tabular}{|c | c | c | c |}
\hline
$\IR$  & $($\#$q\sim1)/($\#$q\sim -\IR)$& mean $p_\R$ injected & exits R
$/$ exits L\\ 
\hline
3 & 2.96& -2.7304  & 1.000031\\
5 & 5.12& -4.5508  & 1.000019\\
7 & 7.29& -6.3712  & 1.000006\\
9 & 9.03& -8.1915  & 1.000005\\
\hline
  \end{tabular}
\end{center}
\caption{Illustration of 1-particle theory. The nominal
    injections are $\IL=1$ and $\IR=3,5,7,9$, and $\alpha=0.1$.
Mean momenta that entered the chain from the right, \ie, 
crossed the rightmost scatterer, are quite close to
$-\IReff = -(1-\alpha )\IR=-0.9 \IR$ (third column). 
   The second column confirms our asserted    
    $\IR : \IL$ ratio of $\IL$ and $-\IR$ in steady states, and the 4th 
    confirms the assertion on left/right exits.}\label{t:table1}
 % \caption{Illustration of 1-particle theory. The nominal
 %   injections are $\IL=1$ and $\IR=3,5,7,9$, and $\alpha=0.1$.
 %   The mean $\IR$ injected (third column) are quite close to 
 %   $(1-\alpha ) \IR=0.9 \IR$. The second column confirms our asserted    
 %   $\IR : \IL$ ratio of $\IL$ and $-\IR$ in steady states, and the 4th 
 %   confirms the assertion on left/right exits.}\label{t:table1}
 % 
 % 
\end{table}

\bigskip
\noindent {\bf Summary.} In the case of a single particle, we find:

\medskip
\noindent (1) The system settles down to an events-time steady state
in which all scatterer energies reflect bath injections with small variations.

\medskip
\noindent (2) In this steady state, we have

(i) $\langle q_i\rangle \equiv 0$, for each $i$ and

(ii) the energy profile is constant along the chain, with  
\begin{equ}\label{e:meanq}
\langle q_i^2\rangle  \ \approx \frac{|\IReff|}{\ILeff+|\IReff|}
(\ILeff^2/\epsilon ^2) + \frac{\ILeff}{\ILeff+|\IReff|}
(\IReff^2/\epsilon ^2)~,
\end{equ}

\medskip
\noindent (3) Momentum transport along the chain is {\it ballistic} in the 
sense that if we view scatterer energies as a sequence of $1$'s and $-3$'s
as defined above, then all the $1$'s move monotonically from left to right
and all the $-3$'s move from right to left.

%%%%%%%%%%%%%%%%%%%%%%%%%%%%%
\subsection{Dynamics of chains with multiple particles}\label{s:lowdensity}

Consider first the case where there is more than one particle in the
chain, but that the particle density, $\rho$, defined to be 
the number of particles per scatterer, is $\ll 1$.
In this case, the dynamical picture can be
described as follows:

Most of the time, no two particles share a gap or find themselves in adjacent gaps,
and the actions of the particles are ``independent'', meaning  
they do as in the single particle case. As a consequence
 the only prominent changes
in scatterer energies occur when a particle enters an interval 
$(i, i+1)$ with $q_i > q_{i+1}$; these values are flipped 
to $q_i' \approx q_{i+1}$ and $q_{i+1}' \approx q_i$
as the particle exits this interval. 

When two or more consecutive intervals 
are occupied by particles, the actions of the particle can interfere 
with one another. One such scenario is illustrated in \fref{f:2rattling}.
In the window of time depicted, the first half of the events show
particle 1 flipping $q_{i-1}$ and $q_i$. Before this flip is complete, 
however, particle 0, which is in the interval $(i, i+1)$, springs into
action, causing $q_{i+1}$ to rise at the expense of the pair $q_{i-1}$
and $q_i$, both of which are pushed down to compensate.

\begin{figure}[h]
  \C{\psfig{file=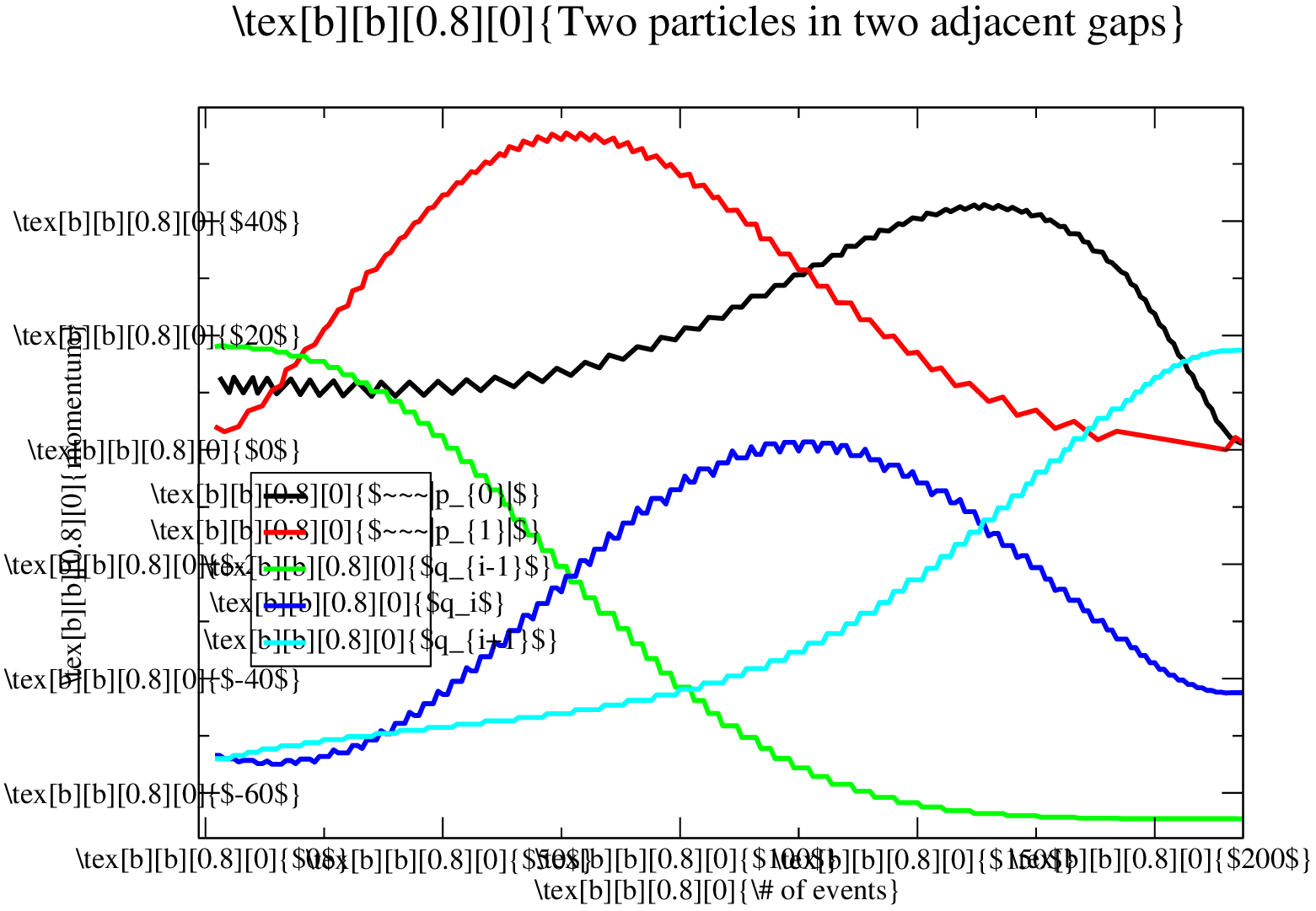,width=0.6\textwidth,angle=270}}
  \caption{The interaction of two particles ($p_0$ (black) and $p_1$
    (red)) with  scatterers at locations $i-1, i$ and $i+1$.
    The horizontal axis is events. The vertical axis is momentum 
    for the $q_j$ and the $|p_k|$. Particle $0$ rattles between $i$ and $i+1$ and
    particle $1$ rattles between $i$ and $i-1$.}\label{f:2rattling}
 \end{figure}

While the kind of interaction in \fref{f:2rattling} can happen, we have found 
that in a good fraction of encounters, 
scatterer momenta are in fact swapped in a {\it sequential manner}, 
leading to a simple reordering of the $q_i$'s. This can be understood
as follows: Suppose a particle with momentum $p_0$ has just entered
an interval bounded by $q_1 \approx 1/\epsilon$ and $q_2 \approx -3/\epsilon$.
In general, $p_0$ can be anywhere between $0$ and $1$ (roughly speaking),
and it takes $1/p_0$ units of time to traverse the interval. With the first reflection, 
$|p|$ starts
to ramp up by ${\cal O}(1)$ per collision, but notice that the effects of these initial 
collisions on the $q$'s are quite insignificant, and that remains true until $|p|$ 
reaches $\sim 1/\epsilon$, because initially $|q| \sim 1/\epsilon$, and 
with each collision, it changes by only $\approx \epsilon |p|$ (see \sref{s:rattling}). 
The same is true as $|p|$ ramps down. What this tells us is that in real time,
the ``flipping" of the $q$'s occurs effectively on a time interval that comprises 
only a fraction of the total
``rattling time"; the smaller the initial $p_0$, the smaller this fraction. 
Thus even when two particles occupy adjacent intervals, if they do not flip
the $q$'s at roughly the same time,  the action is likely to be sequential. 

Baths are much more likely to disrupt this sequential pattern.
Recall that it takes many attempts before a particle from the bath
crosses the first scatterer (\sref{s:bath}). Consider, for example, 
the case where
particle 0 is in the gap $(0,1)$ and particle 1 is in $(1,2)$,
and that particle 1 becomes active while particle 0 is in the middle
of attempting to 
gain entrance. Unlike the situation described in the last paragraph,
the back-and-forth motion of particle 0 between the left bath and 
scatterer 1 is relentless. Its attempt to bring $q_1$ up can easily
interfere with the attempt by particle 1 to flip $q_1$ and $q_2$,
which involves pulling $q_1$ down. 

Baths have been observed to lead to more complicated
scenarios than in \fref{f:2rattling}. We will not enter into such 
a discussion here, but see \sref{s:energy}B.

\medskip
The overall picture can be summarized as follows: 

\medskip
\noindent (1) In {\it low particle-density regimes}, 
say for $\rho < 1$, one can, 
to some degree, extrapolate from the
single-particle picture: 
Many of the $q_i$'s are $1$'s and $-3$'s (see \sref{s:single}), and  
a large majority of the changes in scatterer energies are flips
of the kind discussed above -- except that flips can now occur in 
different parts of the chain in random order. Occasionally, 
simultaneous action in same or adjacent gaps leads to the 
creation of intermediate $q$-values. These values, as with all
$q$-values, are eventually ``flushed out" of the chain. We have found
that among the intermediate $q$-values created, the very small ones 
(we will call them ``zeros'') are the most stubborn: they tend to remain 
in the chain for very long times.\footnote{Zeros form barriers which 
particles cannot cross easily, causing them to oscillate
back and forth in a short stretch of the chain for a long time. The 
mechanism is as follows: 
Assume scatterer momenta are $0,1$ or $-3$,
and focus on a particular $0$. In order for this $0$ to move,
a particle must enter one of the intervals adjacent to it.
If a particle approaches from the left, then the configuration is likely
to be $(1,0)$. The particle then flips $(1,0) \to (0,1)$ and exits 
(in all likelihood) to the right, our $0$ having moved one slot to
the left. Similarly, if a particle approaches from the right, then 
$(0,-3) \to (-3,0)$, the particle exits to the left, and our $0$ moves
to the right. Particle approaches can occur in any order, 
but notice that the {\it net move} of our $0$ cannot exceed the
number of particles in the chain: Suppose the number of particles 
is $k$, and our $0$ is $k$ slots to the left of where it was originally.
The argument above tells us that all $k$ particles must now be to the right of it, so
the next approach is guaranteed to be from the right! This oscillatory behavior
eventually ends when our $0$ acquires ``by accident" a reasonable size
(flips do not return exact pre-flip $q$-values).}

\medskip
\noindent (2) {\it As $\rho$ increases}, 
simultaneous action of particles
in neighboring intervals becomes more commonplace. These actions
may, in principle, enhance or cancel each other, but we see much
more of the latter: {\it pulls in opposite directions by competing forces
lead to decreased amplitudes in the oscillations of $q$-values}. 
The decrease in amplitude on the $-3$ side is somewhat more pronounced, 
possibly due to the fact that it is easier to interrupt a ``longer swing".
This phenomenon is important; it will help explain some of the 
observations in the next subsection. 

\medskip
\noindent (3) In {\it higher particle-density regimes}, \eg, for $\rho \ge 5$,
the process becomes untidy, even chaotic. To gain some intuition, 
we invite the reader to imagine, \eg, for $\rho = 10$,
how $4000$ particles move about in a chain with $400$ scatterers.
These particles have quite different velocities, but with a large
enough number of them present, many will be active simultaneously,
and some of the active 
particles will occupy same or adjacent gaps. Each particle behaves
as though its mission is to put into ``the right order" the momenta of 
the two scatterers at the ends of the interval on which it resides, as was discussed earlier. 
The net effect of all the pulls and tugs together will determine the time evolution of the system. 

\medskip
Snapshots of scatterer energies along a chain for two $\rho$'s
are shown in \fref{f:qprofiles}. 

\begin{figure}[t]
  \C{\psfig{file=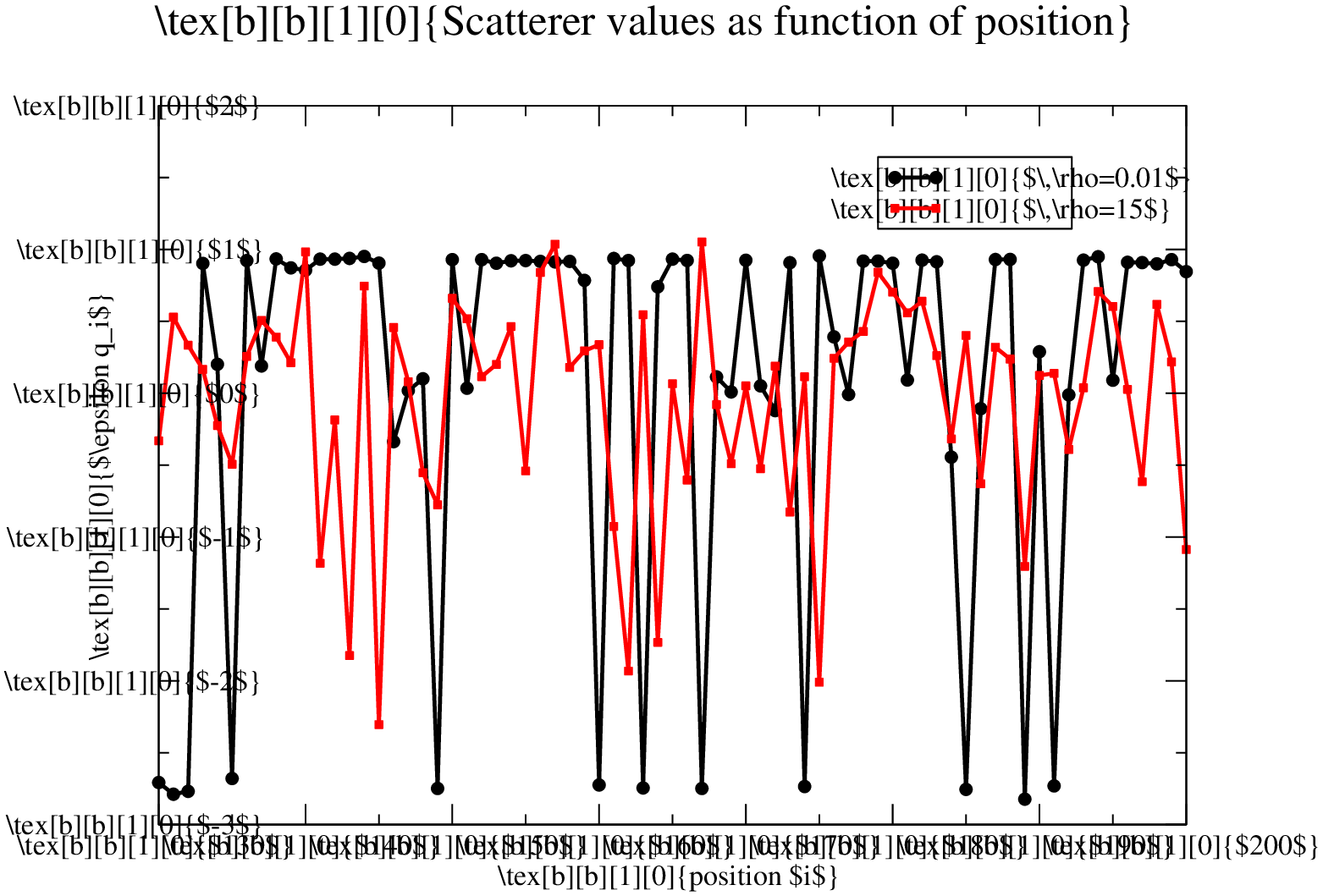,width=0.8\textwidth,angle=270}}
  \caption{A snapshot of scatterer momenta $q_i$ for 70 scatterers
  away from the boundary of a 400-long chain, for 2 different densities $\rho$. 
For $\rho=0.01$, $\epsilon q_i \approx 1, -3$ (injected values)
for many $i$, with a few values near $q=0$.
For $\rho=15$, $\epsilon q_i$-values are more evenly distributed but concentrated 
on an interval strictly smaller than $[-3,1]$; this is due in part to the decreased 
amplitudes of oscillations explained in the text and also to the fact that at any one point
in time, more scatterer momenta are likely to be in the middle of getting ``flipped".}
\label{f:qprofiles}
 \end{figure}

%%%%%%%%%%%%%%%%%%%%%%%%%%%%%%%%
\subsection{Chain energy and fluxes as functions of $\rho$ (and time)}\label{s:energy}

We discuss here two important quantities in relation to energy, namely
the mean total scatterer energy and influx-outflux rates  
into the baths. Since fluctuations are large, to represent these quantities
in a meaningful way, we need to assume there is some stationarity
-- although
as we will see in the next section, that in itself is a point of contention.
Nevertheless, for purposes of the present discussion, we will
show data from simulations averaged over very large windows 
of time, and these windows will increase in size as time progresses. 

\medskip
\noindent {\bf Details on statistics.}
All of our simulations are performed 
for a very large number of scattering events. 
To keep measurements comparable for different densities $\rho$, we define an
  {\bf epoch} to be $1.6\cdot10^9\cdot \rho$ scattering events. There being
  400 scatterers in the chains we use, this
  means that each particle performs on average $4\cdot10^6$
  scatterings per epoch.  Since $\epsilon =0.05$, 
  rattling between two scatterers usually ends after $2\pi/(n\sqrt{2})$ collisions.
  This translates into about $200,000$ crossings per epoch per particle,
  which in turn leads us to expect that
  each particle ``sees'' the bath several hundred times per epoch,
  the exact number depending on whether its motion is closer to
  ballistic or that of a random walk.   Our measurements are averaged
  over each epoch, and we take up to $5000$ epochs.

\bigskip
\noindent {\bf A. Total scatterer energy.} We consider mean total scatterer
energy in the chain. More precisely, for each value of $\rho$
we compute the time average of $\sum_i q_i^2$ in each epoch, with
data taken over many epochs. The results are shown in \fref{f:lsyall}.

Note the agreement with $\rho=0$ (or single-particle) values in \sref{s:single}:
substituting the simulation values into \eref{e:meanq} leads to a
prediction of $\langle q_i^2\rangle \sim 972$, which is quite close to the
limiting value shown. As $\rho$ increases, it is evident that 
total scatterer energy decreases.
The sharpest decrease occurs for $\rho$ between $0$ and $1$.
We have checked that this is due largely to the accumulation 
of larger and larger numbers
of $0$'s (as explained in \sref{s:lowdensity}). As $\rho$ continues to increase,
the decreased amplitudes of oscillations (also explained in
\sref{s:lowdensity})  
will no doubt lead to decreased mean scatterer energy, but there is
another factor that could potentially seriously impact the situation, 
and that is the rates at which energy flows into and out of the chain
at the two ends. We now look at these quantities more closely. 

\begin{figure}[h]
  \C{\psfig{file=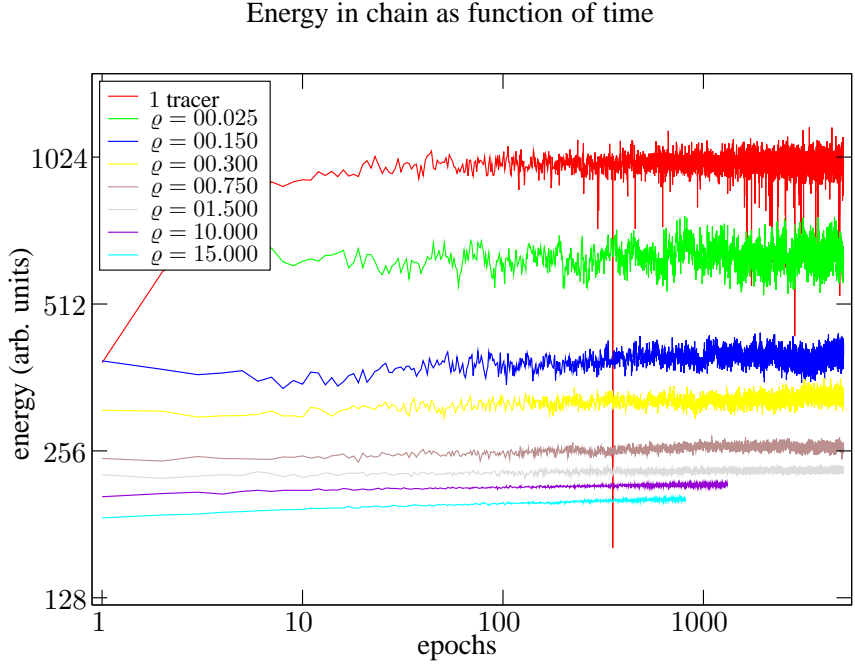,width=0.7\textwidth,angle=270}}
  \caption{Starting from a random initial condition, the mean total
  scatterer energy of a
  chain is measured and plotted as a function of epoch. 
  In time, this energy appears to
  reach some $\rho$-dependent limit, which
  decreases as $\rho$ increases. Closer inspection, however, shows
  that except in the case of a single particle, mean energy drifts
  slowly upwards with epoch.}\label{f:lsyall} 
 \end{figure}

\bigskip
\noindent {\bf B. Energy fluxes.} We consider the following
4 quantities: $\Phi_{\rm L, in}$ and $\Phi_{\rm L, out}$, 
the rates (with respect to physical time) at which energy flows into and out 
of the chain from the left, and $\Phi_{\rm R, in}$ and $\Phi_{\rm R, out}$,
the corresponding quantities on the right. Notice that once we fix the system 
parameters $N=$ number of scatterers, $\rho=$ particle density, 
$\IL$ (resp.~$\IR$) $=$ injected momenta from left (resp.~right) bath,
these fluxes are determined entirely by the system. 

Some simulation results are shown in \fref{f:flux}. 
Observe first that if the system is in a steady (or quasi-stationary) state, 
then by the conservation of energy, one would expect 
$\Phi_{\rm L, in} - \Phi_{\rm L, out} = \Phi_{\rm R, out} - \Phi_{\rm R, in}$,
both quantities being the {\it total flux} $\Phi$ {\it across the system}.
\fref{f:flux} does indeed show these two quantities to be equal. 
The plot shows that activity 
per unit time increases with density, which is expected. But it also
reveals a number of surprising facts. We mention 
two of them:

\smallskip
\noindent 
(1) {\it $\Phi$ as a function of $\rho$.} \  The increase is sublinear. 
To some degree, this can be explained by 
the incomplete flips discussed in \sref{s:lowdensity} and \fref{f:qprofiles}.
The inset in \fref{f:flux} shows, in fact, that $\Phi \sim \rho^\gamma$
for $\gamma \approx 2/3$, which begs for an explanation. 

\smallskip
\noindent 
(2) {\it Contact of system with the right (or hotter) bath.} \  Observe 
that for $\rho =10$ and $15$, $\Phi_{\rm R, out}/\Phi_{\rm R, in} > 2/3$.
The system being in a steady state, particles enter and leave through
the right end of the chain at the
same rate, and each entering particle carries with it an energy 
$\sim 9 = \IR^2$ (see \sref{s:bath}).
This implies that the mean energy carried {\it out} by 
exiting particles is $>6$, and a particle with $p> \sqrt 6$ can exit
only when $q_N$,
 the momentum of the rightmost scatterer, is $> \sqrt 6/\epsilon$ (\sref{s:rattling}). 
Thus the picture cannot resemble that 
 in \fref{f:qprofiles}. 
 
 We conjecture that the dynamics 
next to the hotter bath are somewhat volatile: 
Decreased amplitudes in $q$-oscillations, 
especially on the negative side, make it harder for particles from the right bath 
to enter the system. As they gather between the bath and the $N$th
scatterer, they must exert a nontrivial downward pull on $q_N$. 
Unable to overcome this pull, a ``wall" of large positive $q_i$-values 
builds up to the left of $q_N$ (we have seen this wall many times).
The pressures continue to build, until at some moment $q_N$ swings upwards 
and a floodgate is opened.

\begin{figure}
  \C{\psfig{file=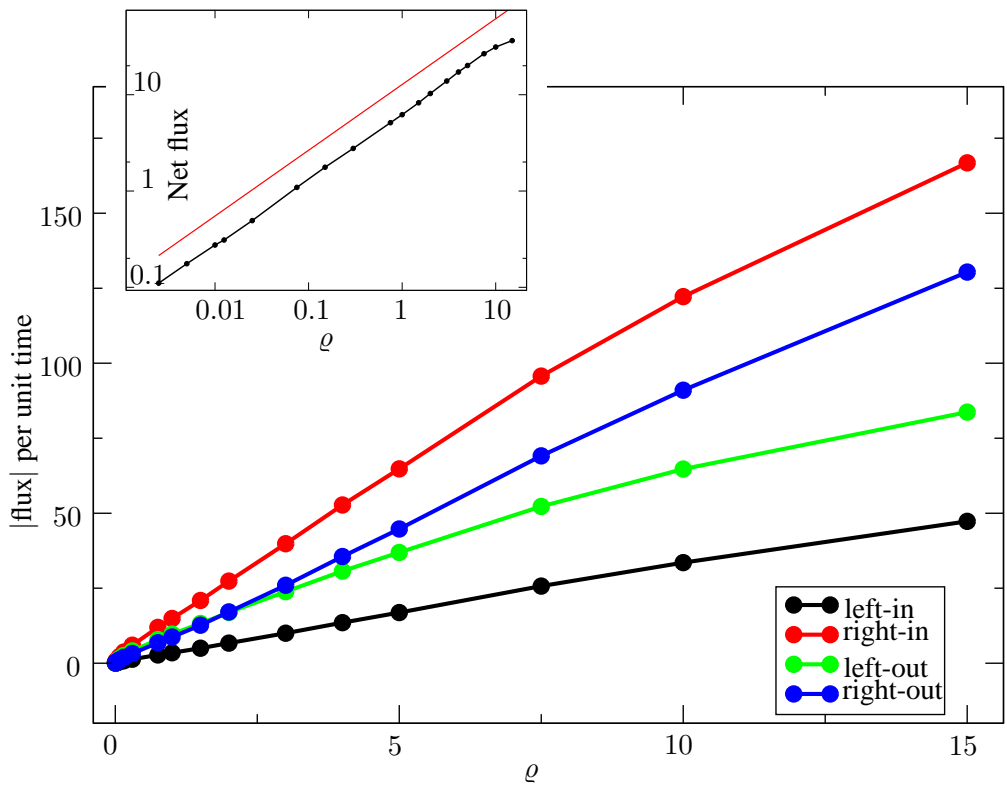,width=0.8\textwidth,angle=270}}
  \caption{The in-flux and out-flux of energy as a function of
  $\rho$. The data show that
    the total flux across the system grows
 sublinearly with the number of particles. The slope of the red line
 is about $0.68$ (inset, log-log). }\label{f:flux}
\end{figure}

%%%%%%%%%%%%%%%%%%%%%%%%%%%%%%%
%%%%%%%%%%%%%%%%%%%%%%%%%%%%%%%%

\section{Freezing}\label{s:freezing}

One of the most salient features of this model is that it never
seems to approach a steady state, in that no matter how long
one waits, many quantities continue to drift, slightly but 
perceptibly. We believe this is due to {\it freezing}, referring to the
slowing down of the particles in the system. There is ample 
mathematical and numerical evidence to support this thinking, 
though we do not have rigorous proofs  and  it is hard for
finite-time simulations to provide conclusive evidence for a 
phenomenon that progresses on a slow (logarithmic) 
scale. In the subsections to follow, we examine the problem from 
several different angles, and attempt to both elucidate and document
this phenomenon. 

\subsection{Distribution of post-collision particle momenta}

To gain insight into  {\it how} freezing 
occurs, we propose to look at distributions of particle momenta 
immediately following collisions with scatterers. 
We begin with a very simple model for which we have a complete
description of the dynamics. 

\subsubsection{Closed system with one scatterer, one particle, and
 two walls}\label{s:onetwo} 

The physical space occupied by this simplified model is the interval
$[0,2]$. At $0$ and $2$, there are no baths but two walls, upon contact
with which the particle is reflected, \ie,  $p'=-p$ where $p$ is the momentum
of the particle. A single scatterer is placed at $1$,
and the interaction of the particle with the scatterer is as before. Thus
the phase variables are $\eta = (q,p,x)$ where $q$ is the momentum
of the scatterer, $q^2+p^2=c^2$ for
some constant $c$ which we may take as $1$, and $x \in [0,2]$ denotes
the position of the particle. The flow is denoted by $\Phi_t$.

Let $\epsilon$ be the constant in the scattering matrix in \sref{s:model}.

\begin{theorem}\label{t:4.1} There is a countable set of $\epsilon$ for which $\Phi_t$ is
periodic. These exceptional values of $\epsilon$ aside, the
following holds:

(i)  For every initial condition $\eta(0)=(q(0), p(0), x(0))$, we have
$$
\frac1T \int_0^T |p(t)|dt \to 0\ .
$$

(ii) The expected time between consecutive scattering events is infinite.
\end{theorem}

\begin{proof} It is advantageous here to focus on
 the {\it collision manifold} $\Sigma = \{x=1\}$, which can be 
 identified with the unit circle
$\{p^2+q^2=1\}$. We let $f: \Sigma \to
\Sigma$ be the first return map, and use the following convention:
Given $\eta = (q,p,1) \in \Sigma$, we first flow, \ie, move right or left
according to whether $p>0$ or $<0$, and do the scattering
when the particle returns to $\Sigma$.
Since $S$ is reflection in $p$ followed by rotation by $\theta$ where
$\theta = -\arctan(\epsilon/\sqrt{1-\epsilon ^2})$,  
$f$ is simply rotation by $\theta$. 

As is well known to be the case, if $\theta$ is rational, then $f$, and
hence $\Phi_t$, is time periodic. This corresponds to a countable set
of $\epsilon$. Outside of this exceptional set of $\epsilon$-values, $f$ is an irrational
rotation, hence it preserves Lebesgue measure $m$ on $\Sigma$ and is
ergodic. Furthermore, all orbits of $f$ are uniformly distributed 
on $\Sigma$. 

We now return to the flow $\Phi_t$, which is what interests us. Let $R: \Sigma \to 
[0,\infty)$ be the first return time to $\Sigma$. Then $R(q,p) = 2/|p|$, 
so $\int R\,dm = \infty$. It follows that 
for every initial condition $\eta(0)$ with no exception, the fraction
of time the trajectory spends in any given neighborhood of $p=0$
tends to $1$ as time goes to infinity. 
\end{proof}

\noindent {\bf Remark 1.} Under the condition that $\theta$ is
irrational, the only invariant probability measures of this system are
$\{\delta_{(1,0,x)}, x \in [0,2]\}$ where $\delta_\eta$ denotes point mass
at the phase point $\eta$. That is to say, all invariant measures are
concentrated at phase points at which the particle is stationary. 
This continuum of
invariant measures taken together can be viewed as a {\it physical
measure} in the sense that given any initial condition $\eta(0)$,
$\frac1T \int \delta_{\eta(t)}dt$ converges to this family as $T \to \infty$.
If one so chooses, this can be used as a mathematical definition of {\it freezing}.

\bigskip
\noindent {\bf Remark 2.} The proof of Theorem~\ref{t:4.1} also tells us 
about the approach to this family of stationary measures, \ie, how 
freezing happens: Fix a very small positive number $p_0$.
The probability of acquiring a value of $p$ with $|p|<p_0$ in a collision
is $\sim p_0$, which is very small, 
but once such a value is acquired, the particle becomes ``inactive"
for $2/|p|$ units of time. It follows that {\it statistically}, 
 the particle spends $100\%$ of its time in these 
extremely-low-energy states, punctuated by (rare) periods of
activity in between.

%%%%%%%%%%%
\subsubsection{Back to the general case}

The analysis of models with $N$ scatterers, $n$ particles and
two baths  is beyond the scope of this paper
(or the state of the art of dynamical systems theory for that matter). 
Taking a hint from the simplified model in \sref{s:onetwo}, 
we  investigate numerically the distribution of momenta 
following scattering events.

Ideally, we fix a system, a particle, say particle $i$,
and an arbitrarily chosen initial condition. 
When particle $i$ exits the system, we name the particle that 
replaces it by the same name. We then make a histogram of the 
momentum of particle $i$ after each scattering event that involves it.
The nature of the distribution as $p \to 0$ contains much information:
If a post-collision momentum is $p$, then the particle 
will be carrying this value of momentum for the next $1/|p|$ 
units of time. Thus a PDF that is roughly constant and bounded away 
from $0$ as $p \to 0$ will imply freezing for the same reason
as before. 

In practice, it is hard to collect sufficient data for individual particles,
so we lump all particles together. 
Results of simulations are shown in \fref{f:pdistribution}.

\begin{figure}
  \C{\psfig{file=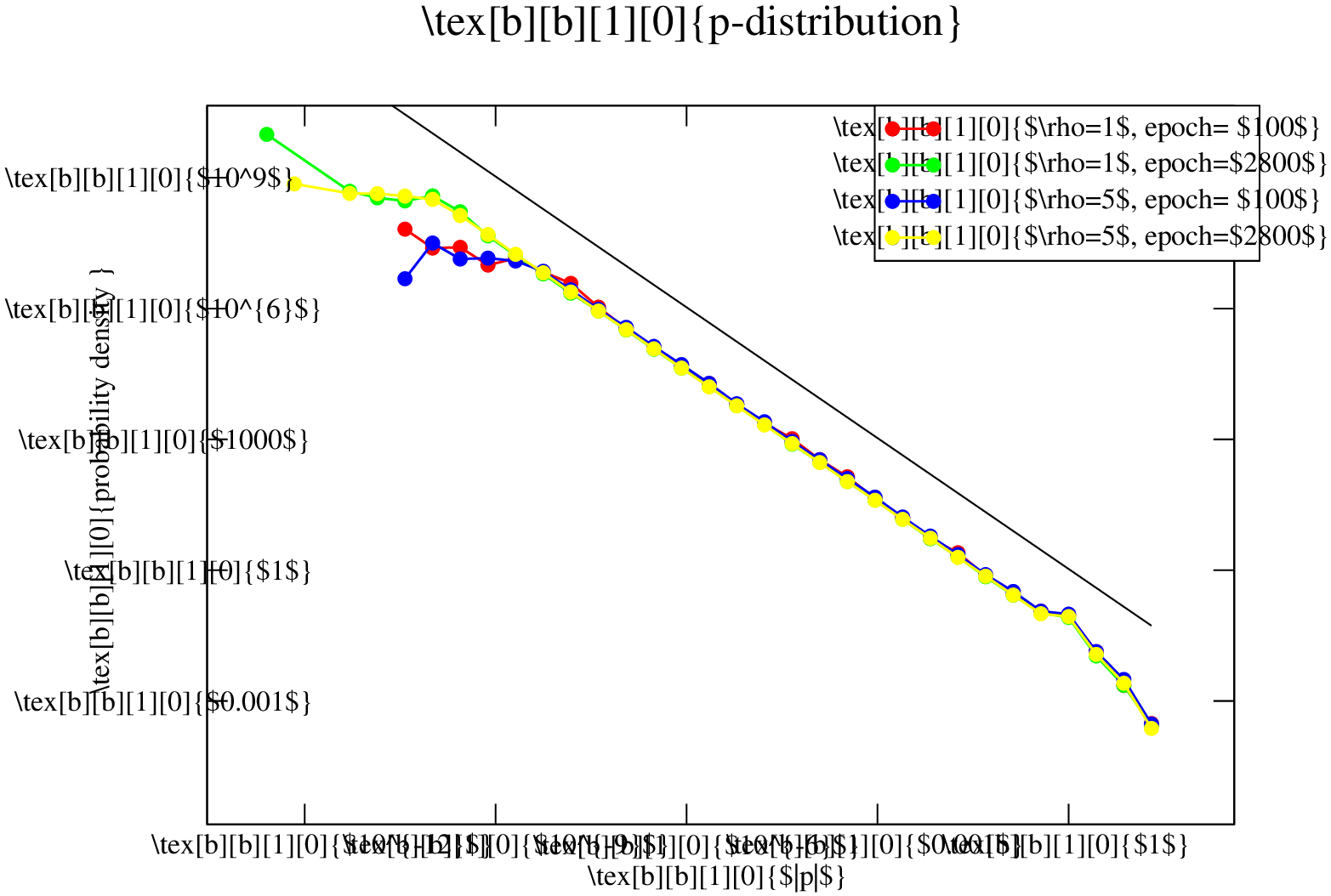,angle=270,width=0.9\textwidth}}
  \caption{Sample distribution of the momenta for  epochs $100$ and $1800$
    for 2 densities (averaged over snapshots $50$ to $100$, resp.~$900$ to $1800$). The black line indicates the slope corresponding
    to density $1/|p|$. Since the time during which a momentum $p$ is
    maintained is $1/|p|$ this demonstrates numerically that the PDF
    averaged over time would be flat near $|p|=0$.}\label{f:pdistribution}
\end{figure}

Notice that the divergence of the integral $\int \frac{1}{|p|} dp$, which lies
at the heart of the freezing phenomenon implied here, relies on
the one-dimensionality of the physical space.
On the other hand, had this model been ``normal", then
when put in contact with a Maxwellian heat bath
of temperature $T$, post-collision momenta should have a PDF
$\sim |p| e^{-\beta p^2}$ where $\beta = \frac1T$, and
should not be $\sim 1$ as $p \to 0$.

%%%%%%%%%%%%%%%%%%%%%%%%%%%%%%%%%
\subsection{Why this system freezes: a theoretical discussion}

In this subsection we seek to understand {\it why} freezing occurs
in our model, and propose an argument that connects 
the no-recoil property directly to energy dissipation in one spatial dimension. Recall that ``no recoil" refers to the fact that the scatterers
do not change their positions (in fact they do not move at all)
even though they carry  ``momenta", and the usual momentum 
exchange between scatterer and particle is assumed
in collisions. 
The argument proposed below  is heuristic, but we believe 
it sheds light on various features of this model and provides 
a theoretical basis for understanding the freezing phenomenon. 

There are two main ingredients in this argument:

\medskip
\noindent (1) Consider the collision of a particle, whose position 
and momentum
are denoted by $x$ and $p$, with a scatterer, the corresponding
coordinates for which are denoted by $y$ and $q$. Had it been normal
Newtonian mechanics,  the volume form
$dxdydpdq$ would be preserved, meaning for an infinitesimal 
volume element 
$V$ in $(x,y,p,q)$-space, $V$ and $\Phi_t(V)$ would have the same
volume assuming the collision occurs between times $0$ and $t$.
Moreover, $\Phi_t(V)$  would be stretched in the $y$-direction whenever 
the scatterer gained energy in a collision. Since assuming 
the scatterer has no recoil is
effectively ``leaving out" the $y$-coordinate, it follows intuitively 
that under this assumption, phase volume (in 3D) is 
contracted when energy is transferred 
from particle to scatterer in a collision.

This intuition is confirmed in
the following direct computation for our model:

\begin{lemma}
The tangent map through a
collision contracts (or expands) the phase space volume by a
factor $|p'/p|$. 
\end{lemma}
\begin{proof}
Consider a little volume $V$ which is at coordinates 
$(p,q,x)$ with $p<0$ and $x<0$ initially, and a time $t$ before which 
all initial conditions in $V$ have made exactly one collision
with a scatterer located at $0$. 
Then at time $t$ we have
\begin{equa}
p'=& -\sqrt{\cdot} p +\epsilon q~,\\
q'=& \epsilon  p +\sqrt{\cdot} q~,\\
x'= & 0 + p'(t-|x/p| )=-\sqrt{\cdot} p +\epsilon q(t-|x/p| )~.\\ 
\end{equa}
The tangent map is then the matrix
\begin{equation}
\begin{pmatrix}
-\sqrt{\cdot}  &\epsilon & 0\\
\epsilon  & \sqrt{\cdot} &0\\
 -\sqrt{\cdot}(t-|x/p| )&-\epsilon (t-|x/p| ) & p'/p  \\
\end{pmatrix}
\end{equation}
and the absolute value of its determinant is $|p'/p|$.
\end{proof}

As we have seen in \sref{s:rattling}, energy transfer can go either way
in collisions,
so in our model, phase-volume is sometimes expanded and sometimes
contracted.

\medskip
\noindent (2) Next look at the problem from a dynamical systems point 
of view.  For a system that expands phase-volume in some parts
of the phase space and contracts it in other parts, when Lebesgue 
measure is transported forward by the dynamics, it is likely to accumulate 
at invariant measures (called {\it physical measures}, see \cite{ER1985}) that are 
volume-contracting, or at least nonexpanding, on average. Intuitively, 
this is the reason why trajectories tend to sinks and not sources, 
attractors rather than repellors. It is connected to the idea of {\it entropy 
production}; see \eg, \cite{Ruellepositivity1996}. There are also 
rigorous mathematical results in the same spirit. For example, it is a mathematical
fact that the sum of all Lyapunov exponents of random diffeomorphisms 
is $\le 0$ always, and is $=0$ if and only if all the constituent maps preserve the same smooth invariant density; see \eg, \cite{Kifer1986}.

\medskip
We now combine the ideas in (1) and (2). Consider first a closed
system, \ie, a system that operates in isolation (and is not connected to baths). 
(2) says that asymptotically in time, volume is likely to be 
decreasing. According to (1), this means
net energy flow is from particles to scatterers.
To summarize, {\it the no-recoil property in our model 
has the following implication: from the view of the
particles, collisions with scatterers lead -- on average --
to energy dissipation.} This is consistent with a slow-down of particles.

In an open system, however, particles are ``recharged" following visits to
baths. Indeed, the expected number of crossings (referring to the 
crossing of scatterers) a particle makes between visits to the baths 
is likely to be finite,  and the process is renewed after each such visit.
Thus the freezing process, which occurs for the same reason as
in closed systems, cannot be completed. 
This reasoning suggests that {\it in open systems the effects of 
the slow-down should be more severe in longer chains}. (It also suggests
that if the transfer of energy from particles to scatterers is 
strictly positive, then the scatterers will heat up eventually, 
though we have not seen any evidence of that in simulations.)

%%%%%%%%%%%%%%%%%%%%%%%%%%%%%%

\subsection{Numerical evidence of freezing}\label{s:numevidence}

We provide here two sets of numerical evidence that document 
the slow-down of particles in chains of length $400$.

\bigskip
\noindent {\bf A. Times between collisions.} Perhaps the most direct
way to document freezing in a chain is to measure the times between collisions as a function of epoch. 
We carried out a study which goes as follows: Start from an initial
condition, run the system for a very long time, and take snapshots 
of the system
at various points in time. At each such time $t$, we asked how long 
(in real physical time) it will be before $X \%$ of the particles have 
engaged in at least one collision. Suppose this happens
at time $t+T(t,X)$. Then the slowing down of the system will be
reflected in the growth of $T(t,X)$ as a function of $t $ for 
fixed values of $X$.
The results for $X=65$ for a number of $\rho$ are shown in \fref{f:corrwaiting35}.
\begin{figure}
\C{\psfig{file=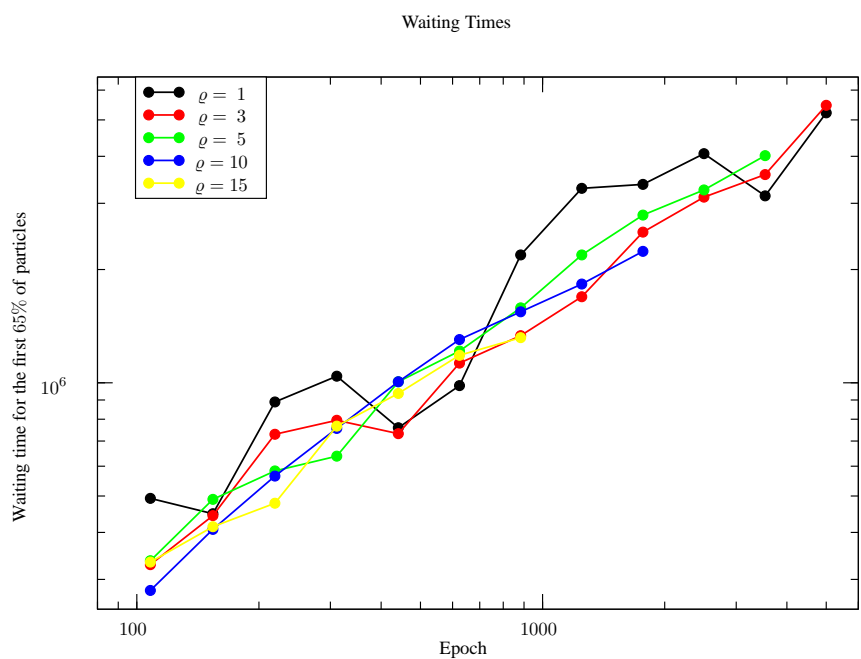,angle=270,width=0.9\textwidth}} 
  \caption{Mean waiting time for $65\%$ of particles as a function of epoch, for several
    densities. Each data point is an average over the epochs from the
    last data point to the current one (and the first from epoch $70$ to
  $100$).}\label{f:corrwaiting35}
\end{figure}

\bigskip
\noindent {\bf B. Energy of a system as function of time.} Consider the
following two facts:

\smallskip
\noindent (1) Since particles
with very low energies do not participate in the evolution of the system
(other than altering their own positions), the {\it effective} density in
 a system should be smaller than its true density
$\rho$. Moreover, with freezing, one would expect
this effective density to decrease with time. 

\smallskip
\noindent (2) We saw  in \sref{s:energy}
that the mean total scatterer energy of a system is larger for systems with smaller $\rho$. 

\smallskip These two facts together should imply that for a given
system with a fixed $\rho$, its mean total scatterer energy taken over different epochs should creep up slowly.
Small inclines in the plots in \fref{f:lsyall} are indeed evident.

\bigskip
\noindent {\bf Important remark.} While much of this section is focused
on freezing,  it is just as important to remember that 
the process occurs  extremely slowly: Times between collisions, as we have 
seen in \fref{f:corrwaiting35}, increase
only on a $\log \log$-scale. Following an initial transient, 
the models considered in this paper remain
in a {\it slowly-varying, quasi-stationary state} for a 
very long time -- indefinitely for practical purposes. 
For this reason, we submit that
macroscopic observations such as those in \sref{s:energy} 
are entirely meaningful, even though no well defined limiting 
values can be reached
in finite time.

%%%%%%%%%%%%%%%%%%%%%%%%%%%%%%

\subsection{The final state}\label{s:final}

The following theoretical question begs for an answer, however:
What happens as time goes to infinity for a chain that is
arbitrarily long? In what follows, let us 
consider {\it events time}, skipping over the long periods with no 
collisions occurring anywhere in the chain. 

\medskip
\noindent {\bf Solo-particle conjecture:} {\it For long
periods (in events time), the baton is carried by a single particle,
\ie, it rattles between scatterers and 
walks about in the chain while all other particles are ``asleep".
These long solo walks are punctuated by (rare) periods during which multiple particles are simultaneously active. The baton  is passed from particle to particle during these
periods of activity.}

\medskip
The following is the basis of this conjecture: Fix $\epsilon_1 \ll 
\epsilon_2 \ll  1$. Let us say a particle is
``inactive" when $|p|< \epsilon_1$, ``active"
when $|p|> \epsilon_2$. After the system has been running for a long
enough time, each particle will be inactive nearly $100\%$ of 
{\it real} time, and active for essentially $0\%$ of real time. 
Unless there are hidden correlations (we have
no reason to believe there are any), the active
periods of a particle will most likely occur when all other particles
are inactive. Notice also that for $\epsilon_2$ small enough,
an active particle is likely to remain active for many collisions in 
a row assuming $p$-values following crossings are more or less random 
(see \sref{s:rattling}). If the two assumptions above are essentially
 valid, one will see, in {\it events
time}, long walks by solo particles.

Observe that during these long walks, the dynamics will, to some 
degree, resemble those in \sref{s:single}, in that the active particle 
will be ``flipping" scatterer energies wherever it goes (though we
do not have any basis for speculating
if the $q_i$'s will be mostly $-3$'s and $1$'s).

The conjecture above is obviously not testable, nor can any credible
evidence be produced in reasonable time.
We offer some data nevertheless on the distribution of
{\it lengths of runs}, which are defined as follows:
Let a system evolve and record the sequence of crossings of scatterers
as $j_1, j_2, \dots$ where $j_k = i$ if the $k$th crossing
involves particle $i$. (We have elected to count {\it crossings}
rather than events, and failed attempts to enter the chain 
are not counted as crossings.) Define a {\it run} to be
a maximal sequence of consecutive $j_k$'s with the same index.
Thus a run of length $n$ in a system with $k$ particles means
that one of the $k$ particles carries out a walk that comprises 
$n$ crossings before a crossing is made by one of the other $k-1$
particles.

As discussed earlier, we expect these runlengths to get longer
with time. A numerical experiment with 5 particles was carried out, 
and some of the results are presented in \fref{f:slowing}. 
The statistics were delicate due to the occurrence of extreme events. 
In the longest runs observed, a single active particle
made more than $1.5\cdot 10^8 $ crossings before any of the
other $4$ particles crossed a scatterer.

\begin{figure}[ht]
  \C{\psfig{file=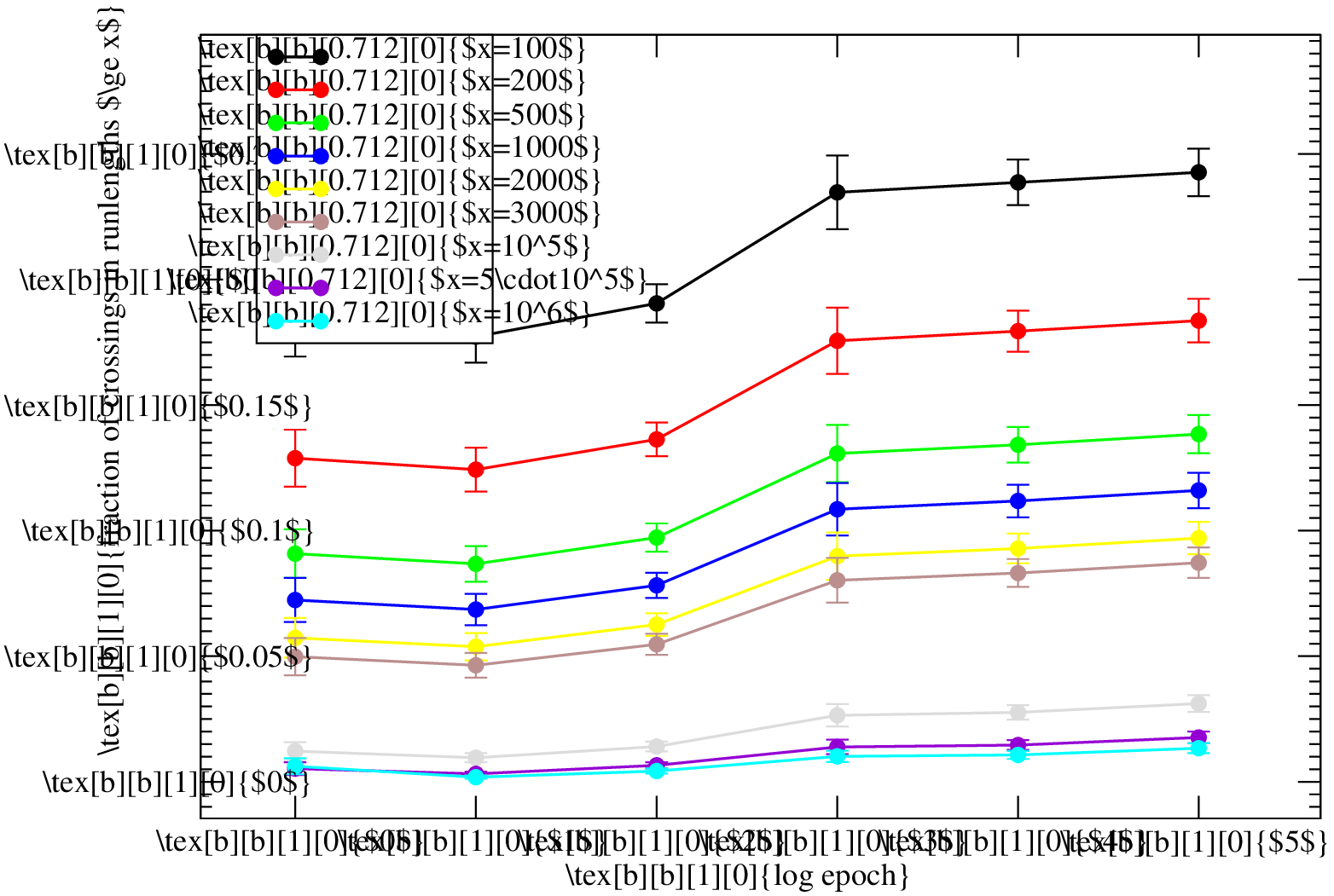,width=0.7\textwidth,angle=270}}
  \caption{Slowing down of particles, for a simulation with $5$
    particles and $600$ scatterers and the standard heat baths ($\sim
    1$, $\sim -3$). Shown is the percentage of crossings spent
    in runlengths longer than $100$, $200$, $\dots$ as a function of
    elapsed time (on a logarithmic scale, $\log=5$ corresponding to
    $8192$ epochs of $2\cdot 10^7$ crossings of scatterers, each). The
    figure is obtained by averaging over $9$ realizations.The
  upward trend of the runlength is clearly visible.}\label{f:slowing}
 \end{figure}

%%%%%%%%%%%%%%%%%%%%%%%%%%%%%%%%%%%%%%%%
%%%%%%%%%%%%%%%%%%%%%%%%%%%%%%%%%%%%%%%

\section{Summary and conclusion}

This papers contains an analysis of a model of heat conduction introduced 
in \cite{CE2009}. Our aim was to deduce from the interaction in \cite{CE2009}
the large-time dynamical behavior of the system and to leverage that
to shed light on physically relevant quantities and phenomena.

Consider first the simplest case of a single particle. We showed that 
the particle rattles between two adjacent scatterers, systematically moving 
momentum from one to the other -- leaving to work on another pair of 
scatterers after completing its task. If one ignores physical time (which is
not meaningful here), this translates into monotonic or ballistic transport of 
momenta along chains, accompanied by oscillatory motion in scatterer momenta.  

When more particles are present, the dynamics can be seen as the sum total of 
individual actions, the net effect of a constant tug of war in which each particle 
moves bits of (positive) momentum from the scatterer on its left to the scatterer 
on its right. This way of viewing the microdynamics enables one to explain readily
certain physical observations that are not obvious otherwise, such as 
the decrease in chain energy as measured by the scatterers
as particle density increases.
(Reason: competing actions of particles in adjacent intervals lead to decreased
amplitudes in the oscillation of scatterer momenta.)

But the most striking feature of this model by far is that it {\it freezes}, 
by which we refer to the slowing down of the particles. Starting from any 
initial condition, one finds that as time goes on, more and more particles 
will become inactive, \ie, they are in very low energy states, 
leading to fewer and fewer collisions per unit time. On the level of 
macroscopic observations, this phenomenon manifests itself in slowly
drifting measurements, independently of the sizes of the epochs used in
statistical averages. (For example, scatterer energies drift up slowly,
consistent with the fact that effectively there are fewer 
particles in the chain.) This drifting,
however, is very slow, as is the freezing process.
We think it is appropriate to view the system as being in a {\it quasi-stationary} state.

Finally, we address what we believe is the root cause of the particle slow-down. 
To simplify the local dynamics, the authors of \cite{CE2009} fixed the positions
of the scatterers while retaining the usual rules of energy and momentum 
exchanges. One way to put it is that the scatterers have {\it no recoil}. 
In such a system, phase-space volume is not conserved by the dynamics. 
Specifically, in a particle-scatterer collision, phase-space volume is contracted
when energy flows from particle to scatterer, expanded when it flows in the 
opposite direction. Since dynamical processes have
a way to head toward volume-contracting regimes (\ie, attractors, and not
repellors), the tendency here is for particle energy to be dissipated.
One is thus reminded again that interactions 
that do not conserve phase volume have consequences.

\noindent{\bf Acknowledgments}
JPE thanks the Courant Institute, where this work was begun, for its
hospitality. JPE was supported in part by the Fonds National Suisse,  LSY was
supported in part by NSF Grant DMS-0600974.

\bibliographystyle{JPE}
\bibliography{refs}

\end{document}